\documentclass[%
 aip,
 amsmath,amssymb,
 reprint,%
]{revtex4-2}

\usepackage{graphicx}
\usepackage{bm}
\usepackage{amsthm}

\theoremstyle{plain}
\newtheorem{theorem}{Theorem}
\newtheorem{proposition}[theorem]{Proposition}
\newtheorem{lemma}[theorem]{Lemma}

\theoremstyle{definition}
\newtheorem{remark}[theorem]{Remark}

\DeclareMathOperator{\Tr}{Tr}

\DeclareMathOperator{\Ran}{Ran}
\newcommand{\Om}{\Omega}
\newcommand{\Vq}{V_{\mathrm q}}
\newcommand{\Vc}{V_{\mathrm c}}
\newcommand{\Nq}{N_{\mathrm q}}
\newcommand{\half}{\tfrac12}
\newcommand{\Jc}{J}
\newcommand{\R}{\mathbb R}

\newcommand{\Q}{\mathbb Q}

\begin{document}

\preprint{AIP/123-QED}

\title{Quantum-advantage resource of a two-mode Gaussian state: Analytical theory of convex optimization and a Galois no-go for the closed-form solution}

\author{Kunwar Kalra}
\email{kunwar.kalra@tamu.edu}
\affiliation{Department of Mathematics, Texas A\&M University,
College Station, Texas 77843, USA}
\author{V.~V.~Kocharovsky}
\email{vkochar@physics.tamu.edu}
\affiliation{Department of Physics \& Astronomy, Texas A\&M University, College Station, Texas 77843, USA}

\date{\today}

\begin{abstract}
We study the problem of extracting a quantum complexity resource from a mixed Gaussian state of the multimode light. 
We present the first complete, certificate-checked solution to this problem in a genuinely coupled sector. 
We carry this out for the two-mode case, the smallest case in which modes are genuinely coupled. 
Even in this case the solution is highly nontrivial, and we rigorously prove that it cannot be given in a closed form. 
\end{abstract}

\maketitle

\section{Introduction}\label{sec:intro}

Multimode, squeezed and entangled, light produced via parametric down conversion or four wave interaction by means of the optical parametric amplifiers (OPAs), oscillators or nonlinear waveguides and interferometers is central to modern quantum optics science and technology. Optical modes constitute a quantum system of harmonic oscillators described by continuous variables (CVs) of their coordinates and momenta corresponding to quadrature operators of the modes' electric field which are related to the annihilation and creation operators of photons as follows $\hat{q}_k = (\hat{a}_k^\dagger + \hat{a}_k)/\sqrt{2}, \hat{p}_k = i(\hat{a}_k^\dagger - \hat{a}_k)/\sqrt{2}$. 
The most interesting and promising features and applications of CV quantum systems in universal fault-tolerant quantum computing, quantum communications, networking, sensing, metrology, cryptography, etc. are associated with their so-called {\it quantum advantage}, or supremacy, over classical systems and computers. 
Quantum advantage of the multimode light originates from $\sharp$P-hard computational complexity of a joint probability distribution of photon numbers in different modes. It fully manifests itself when, instead of probabilities associated with CVs in the phase space, quantum statistics of discrete variables, that is photon numbers, is involved. 

Such an analysis culminated recently in establishing a quantum complexity resource responsible for quantum advantage of the multimode light. It was done for Gaussian boson sampling in the paper \cite{Oh2024} by Oh and co-authors who found an algorithm that simulates classically the joint quantum statistics of photon numbers in the noisy multimode light, but only if the quantum complexity resource residing in the multimode light is not too large. The dimension of the resource is determined by the size of a matrix under the hafnian which gives photon-counting statistics in accord with the hafnian master theorem established in \cite{LAA2022,PRA2022}. The point is that the hafnian is $\sharp$P-complete for computing \cite{Barvinok2016} and provides a universal tool for analysis of computational $\sharp$P-hardness since, according to Toda's theorem \cite{Toda1991,Basu2012}, there is a deterministic polynomial-time Turing reduction of any problem in the polynomial hierarchy to a counting problem relative to a hafnian oracle. 
Based on this generality, establishing it as a legitimate measure of multimode light's quantum complexity in the quantum-information sense, we name it {\it the quantum-advantage resource} \cite{OPAarXiv2026,Entropy2026}.  

The most well developed and widely employed sources of multimode light in modern studies and applications of quantum optics science and technology are OPAs which, as is well known, generate quantum light in Gaussian states \cite{Serafini2023}. They also are employed in the most advanced recent setups aimed at demonstrating a prototype of a universal fault-tolerant photonic quantum computer 'Aurora' \cite{RadNature2025,Takeda2019} and quantum advantage in Gaussian boson sampling with a linear interferometer \cite{Pan2026,PanPRL2021,Deshpande2022,Madsen2022,Deng2023,Hamilton2017}. In fact, multimode light in a Gaussian state is capable of full $\sharp$P-hardness and quantum advantage. This is the reason why the above resource was introduced and studied for Gaussian states \cite{Oh2024,OPAarXiv2026,Entropy2026}. The present paper also deals with the Gaussian states. A generalization to the non-Gaussian states requires additional analysis as is explained in \cite{OPAarXiv2026}.  

A Gaussian state of $M$ bosonic modes (for instance, optical modes carrying squeezed light) is described, up to a displacement we may set aside, by a real symmetric $M\times M$ \emph{quadrature covariance matrix}  
\begin{equation} \label{VG}
V = \left[ \begin{matrix}
\langle \hat{p}_k \hat{p}_{k'} \rangle 
            &   
\frac{1}{2}\langle \hat{q}_k \hat{p}_{k'}+\hat{p}_{k'} \hat{q}_k \rangle^T
            \\[6pt]
\frac{1}{2}\langle \hat{q}_k \hat{p}_{k'}+\hat{p}_{k'} \hat{q}_k \rangle 
            &   
\langle \hat{q}_k \hat{q}_{k'} \rangle
        \end{matrix} \right] 
= \frac{i}{2}\Omega + \langle \hat{s}\hat{s}^T \rangle. 
\end{equation}
This matrix collects the variances and correlations of the quadrature operators of the modes, listed in the vector $\hat{s}= (\hat{p}_1,\dots,\hat{p}_M,\hat{q}_1,\dots,\hat{q}_M)^T$. These operators do not commute. They obey canonical commutation relations $ [\hat{q}_k,\hat{p}_{k'}]=i\delta_{kk'}$ associated with the \emph{symplectic form} 
\begin{equation} \label{Omega}
\Om=\begin{pmatrix}0&\mathbb{I}_M\\-\mathbb{I}_M&0\end{pmatrix}.
\end{equation}
Heisenberg's uncertainty principle, in its multimode (Robertson--Schr\"odinger) matrix form, states that a real symmetric $V$ is the covariance matrix of an actual quantum state, in which case we call it \emph{physical}, exactly when the Hermitian matrix $V+\tfrac{i}{2}\Om$ is positive semidefinite, $V+\tfrac{i}{2}\Om\succeq0$. The least uncertain state, the vacuum (no photons), sits at $V=\half \mathbb{I}_M$.

To isolate the genuinely nonclassical content of the covariance, $V$, Oh and co-authors introduced~\cite{Oh2024} a semidefinite program (SDP), that is, a convex optimization over positive-semidefinite matrices, that splits $V$ into the classical part $V_c$ and the smallest physical \emph{quantum} part (the quantum complexity resource) $V_q$ which still leaves a legitimate \emph{classical} noise remainder:
\begin{equation}\label{eq:sdp}
\Vq^\star=\arg\min\bigl\{\,\Tr(\Vq)\ :\ \Vq+\tfrac{i}{2}\Om\succeq0,\ \
V-\Vq\succeq0\,\bigr\},
\end{equation}
$$\qquad \Nq(V)=\half\bigl(\Tr\Vq^\star-2\bigr).$$
The first constraint asks the quantum part $\Vq$ to be physical in its own right; the second asks the remainder $\Vc=V-\Vq$ to be a positive matrix, which is precisely the condition for it to be admissible classical noise (it then represents a random classical displacement, an operation that is free to simulate). Among all such splittings the program selects the one of least trace. Because the trace of a covariance matrix is a measure of total energy, minimizing it concentrates all irreducible nonclassicality in $\Vq^\star$ and assigns the rest to classical noise. The associated photon number $\Nq(V)=\half(\Tr\Vq^\star-2)$ counts photons above the vacuum (the trace of vacuum covariance is 2), and it is a certified lower bound on the classical cost of sampling the state. Locating $\Vq^\star$ is therefore the basic problem. 

We recall that the minimum of the functional $\Tr(\Vq)$ subject to the constraints in Eq.~(\ref{eq:sdp}) can be found by the method of Lagrange multipliers, familiar from Lagrangian mechanics, quantum field theory, and statistical physics. It provides first-order equations determining the minimum in question via the saddle point of the Lagrangian function with respect to variation of the original variables and Lagrange multipliers without necessity to take care of the constraints separately. An example, closely related to the multimode system in question, is introducing the term $\mu \hat{N}$ via the chemical potential (Lagrange multiplier) $\mu$ into the effective Hamiltonian to switch from the canonical ensemble, constrained by the particle number conservation condition $\hat{N} = \rm{const}$, to the grand canonical ensemble in the quantum many-body theory of Bose-Einstein condensation \cite{Zubarev1974,JStatPhys2015}. 

We follow the well known Karush–Kuhn–Tucker (KKT) generalization of the Lagrange multiplier method for convex optimization which contains constraints in the form of inequalities rather than just standard Lagrange equalities \cite{BenTalNemirovski}. The only additional complication of the convex optimization (\ref{eq:sdp}) is that it deals not with the usual functions of real-valued variables but with the functional $\Tr(\Vq)$ varied over the set of real symmetric matrices $\Vq$. So, the Lagrange multiplier in our case is also a matrix $S$, and it appears in the generalized Lagrangian via an inner product, $\Tr(S(V-\Vq))$, with the matrix $V-\Vq$ subject to the inequality constraint $V-\Vq\succeq0$ in Eq.~(\ref{eq:sdp}). Here $X\succeq Y$ denotes the L\"owner order, meaning that $X-Y$ is positive semidefinite.   

The goal of the paper is to \emph{determine the resource $\Vq^\star$ explicitly} for the system of two modes, $M=2$. Two modes form the first case in which the modes genuinely couple. For a single mode the answer is the squeezed-vacuum solution \cite{Oh2024,Entropy2026}. The explicit solution to the two-mode quantum-advantage resource problem is the continuous-variable counterpart of the explicit two-mode squeezed-state solution \cite{Serafini2023,Barnett1996,Weedbrook2012} familiar from quantum optics. The complete analytical theory of the two-mode quantum-advantage resource presented in this paper shows that the solution to the two-mode problem is already nontrivial, not as simple as the solution to the two-mode squeezed state problem, and actually cannot be given in the closed form. This explains why the solution was not found earlier. 

The general multimode prescription, building $\Vq^\star$ by ``nullifying the sub-vacuum eigenvalues'' of $V_c$, was left open in our previous paper~\cite{Entropy2026} (problem (iii) there).

\subsection*{Relation to prior work}

We build on the structural results of Ref.~\cite{Purity}, referred to throughout as \emph{Purity}, which establishes, for any number of modes, that the optimizer of~\eqref{eq:sdp} exists, is unique, and is pure, supplies the
closed-form oracle and Riccati identity for the inner problem that we use as our main computational tool, and solves the passive-diagonalizable class in closed form. The present paper carries the two-mode case through the first genuinely coupled sector, where the passive branch of \emph{Purity} no longer applies, and isolates the exact algebraic obstruction to a closed-form solution there; the general multimode problem is left to future work. The 
underlying resource interpretation of $\Nq$ is from~\cite{Oh2024,OPAarXiv2026,Entropy2026}.

\section{The optimizer and the language of duality}\label{sec:background}

\subsection{Results used from \emph{Purity}}\label{sec:purity}

We use two facts established in \emph{Purity}~\cite{Purity}. We state exactly which results are used and define each notion as it appears.

First, \emph{the minimizer of \eqref{eq:sdp} exists, is unique, and is a pure state.} A pure Gaussian state is one with a definite wavefunction, carrying no residual classical mixing; at the level of covariance matrices it saturates the uncertainty relation, all of its symplectic eigenvalues taking the minimum value $\tfrac12$. Every pure two-mode covariance arises from the vacuum $\half \mathbb{I}_4$ by a passive rotation followed by squeezing,
\begin{equation}\label{eq:pure}
\Vq=\half\,O^\top\!\operatorname{diag}\!\bigl(e^{2r_1},e^{2r_2},e^{-2r_1},
e^{-2r_2}\bigr)O,
\end{equation}
$$\qquad O\in\mathrm U(2):=\mathrm O(4)\cap\mathrm{Sp}(4),\ \ r_1,r_2\ge0 .$$
A \emph{passive} transformation $O$, built from beam splitters and phase shifters, conserves photon number and is represented on the quadratures by a matrix that is simultaneously orthogonal and symplectic; such matrices form the group $\mathrm U(2)=\mathrm O(4)\cap\mathrm{Sp}(4)$. The numbers $r_1,r_2\ge0$ are the \emph{squeezing} parameters: along one quadrature of mode $i$ the variance is reduced below the vacuum value to $\tfrac12 e^{-2r_i}$, while its conjugate is increased to $\tfrac12 e^{2r_i}$, with the product held at $\tfrac14$ as the uncertainty relation requires. The eigenvalues of $\Vq$ are accordingly the
reciprocal pairs $\tfrac12 e^{\pm2r_i}$. Because $O$ is orthogonal it leaves the trace unchanged, so the trace depends only on the squeezing,
\begin{equation}\label{eq:framefree}
\Tr(\Vq)=\cosh2r_1+\cosh2r_2 ,
\end{equation}
and \eqref{eq:sdp} is the problem of using the least squeezing for which $\Vq$ still lies below $V$:
\begin{equation}\label{eq:var}
\min_{O\in\mathrm U(2),\,r_{1,2}\ge0}\ \cosh2r_1+\cosh2r_2
\quad\text{subject to} 
\end{equation}
$$O V O^\top\succeq\half\operatorname{diag}(e^{2r_1},e^{2r_2},e^{-2r_1},e^{-2r_2}).$$

Second, \emph{Purity solves the inner problem in closed form through an ``oracle''.} For any fixed positive-definite weight matrix $A\succ0$, the pure state that minimizes the weighted trace $\Tr(A\Vq)$ over all physical $\Vq$ is given explicitly by
\begin{equation}\label{eq:oracle}
\Vq(A)=\half\,A^{-1/2}\bigl|A^{1/2}\Om A^{1/2}\bigr|A^{-1/2},
\end{equation}
$$\qquad \Vq(A)\,A\,\Vq(A)=\tfrac14\,\Om^\top\!A\,\Om ,$$
where $|M|:=\sqrt{M^\top M}$ denotes the matrix absolute value. We call the map $A\mapsto\Vq(A)$ the \emph{oracle}: for each weighting of the quadratures it returns the pure state of least weighted trace. The identity on the right of \eqref{eq:oracle} is a \emph{Riccati identity}. A Riccati equation is a quadratic matrix equation, one in which the unknown $X$ enters to second order through a product $X A X$ (the name is after Jacopo Riccati, and such equations are central to optimal control and filtering). Here it states that the oracle output is exactly the pure $X$ solving $X A X=\tfrac14\Om^\top A\,\Om$. The minimum value attained equals $\Tr\bigl(A\,\Vq(A)\bigr)=\nu_1(A)+\nu_2(A)$, the sum of the two \emph{symplectic eigenvalues} of $A$. The symplectic eigenvalues of a positive matrix $A$ are the positive numbers $\nu_i$ for which $\pm\nu_i$ are the eigenvalues of $i\Om A$. They are the invariants of $A$ under symplectic (canonical, commutation-preserving) changes of frame, and for a physical covariance they are exactly the quantities the uncertainty relation constrains to be $\ge\half$.

\subsection{Certifying optimality: the dual problem}
\label{sec:duality}

The program \eqref{eq:sdp} is a constrained minimization, called the
\emph{primal}. Attached to it is a second optimization, the \emph{dual}, assembled from the constraints; its variable here is a positive-semidefinite matrix $S\succeq0$, a Lagrange multiplier for the matrix inequality $V-\Vq\succeq0$. The dual is a maximization, and \emph{weak duality} is the elementary fact that every value of the dual is a lower bound for every value of the primal. A \emph{dual certificate} is a choice of the dual variable for which this lower bound equals the value of a candidate primal
solution. Producing one therefore \emph{proves} that the candidate is optimal, with no further search needed. In the present setting the criterion takes a concrete form. Writing the weight as $A=I+S$, a feasible pure state $\Vq$ is optimal as soon as there exists some $S\succeq0$ with
\begin{equation} \label{S}
\Vq=\Vq(I+S)\qquad\text{and}\qquad \Ran(S)\subseteq\ker(V-\Vq).
\end{equation}
The first equation says $\Vq$ is the oracle's answer for the weight $I+S$. The second is \emph{complementary slackness}: the multiplier $S$ may act only where the constraint $V-\Vq\succeq0$ is tight, that is, on the kernel of $V-\Vq$, the set of directions along which no classical noise remains. This certificate is exhibited explicitly in each closed-form case below.

That such a certificate exists at all is guaranteed by a mild nondegeneracy assumption. We take $V$ to be \emph{strictly mixed}, meaning all of its symplectic eigenvalues exceed $\half$. This is \emph{Slater's condition}, the standard requirement that the feasible set have a genuine interior point (here, that some quantum state fit strictly inside $V$); it ensures that the primal and dual optima coincide and that the dual optimum is attained, so a certificate is available. The boundary case, in which $V$ is itself already pure and the feasible set degenerates to the single point $\Vq^\star=V$, is recovered by continuity. Throughout, the closed forms derived below are checked against a direct numerical solution of the primal SDP, and the tolerances quoted come from those checks.

\section{The nonclassical directions and the preimage criterion}\label{sec:tools}

The program acts only on the directions in which $V$ is more certain than the vacuum. Diagonalize the real symmetric matrix,
$V=\sum_{j=1}^4\lambda_j\,v_jv_j^\top$, with orthonormal eigenvectors $v_j\in\R^4$ and eigenvalues $\lambda_j$. Call a direction $v_j$ \emph{nonclassical}, or \emph{sub-vacuum}, when its variance dips below the vacuum floor, $\lambda_j<\half$, and let $\kappa$ be the number of such directions. 
They correspond to the directions of the minor axes of Wigner quasi-probability distribution in the phase space and set the universal lower bound for quantum advantage established in \cite{Entropy2026}.
Only these directions matter: on the orthogonal complement, where $V\succeq\half I$, the vacuum already fits beneath $V$ and no quantum part is extracted there.

A single geometric operation organizes the entire analysis: the quarter-turn phase rotation that exchanges position and momentum, $q\mapsto p,\ p\mapsto-q$. On the quadratures it is the orthogonal matrix
\begin{equation}\label{eq:J}
\Jc=\begin{pmatrix}0&-I_2\\ I_2&0\end{pmatrix},\qquad \Jc^2=-I_4 .
\end{equation}
Since $\Jc^2=-I$, it acts like multiplication by $i$: it is a \emph{complex
structure} on the real quadrature space (a real linear map whose square is $-I$, which lets one regard the real space as a complex one). It sends each direction $v$ to a perpendicular partner $\Jc v\perp v$, its \emph{conjugate quadrature}. The reciprocal-pair structure of a pure state \eqref{eq:pure} is precisely the statement that its eigenvectors come in conjugate pairs $(v,\Jc v)$ with
reciprocal eigenvalues $(\rho,\tfrac1{4\rho})$, because the passive $O$ commutes with $\Jc$.

Beyond the oracle we use its converse: a test for when a prescribed pure state is the oracle's output. This test underlies every optimality proof below.

\begin{lemma}\label{lem:preimage}
Let $X$ be a pure covariance and $A\succ0$. Then $\Vq(A)=X$ if and only if the Riccati identity $XAX=\tfrac14\Om^\top\!A\,\Om$ holds, equivalently if and only if $A\Om$ commutes with the complex structure $\Jc_X:=2\Om X$. Consequently a pure feasible $\Vq$ is the optimizer of~\eqref{eq:sdp} if and only if there is an
$S\succeq0$ with $\Ran(S)\subseteq\ker(V-\Vq)$ and $\Vq(I+S)=\Vq$.
\end{lemma}

The content of the lemma is an inversion. The oracle \eqref{eq:oracle} sends a weight $A$ to a pure state $\Vq(A)$; the lemma tells us exactly which weights $A$ return a prescribed target $X$. Recast through the matrix $\Jc_X=2\Om X$, which is again a complex structure ($\Jc_X^2=-I$), the condition becomes a
single commutation relation. Combining this preimage test with the complementary-slackness requirement gives a finite, directly checkable optimality criterion, and the three closed-form theorems below all proceed by exhibiting the certificate $S$ explicitly.

\begin{proof}
The equivalence of $\Vq(A)=X$ with the Riccati identity
$XAX=\tfrac14\Om^\top\!A\,\Om$ is the inversion of the oracle proved in \emph{Purity}: among pure states, $\Vq(A)$ is the unique solution $X$ of that quadratic equation. It remains to rewrite the identity as a commutator. Because $X$ is pure it is invertible and satisfies the purity relation $X\Om X=\tfrac14\Om$, so $\Jc_X:=2\Om X$ obeys
\[
\Jc_X^2=4\,\Om X\Om X=4\,\Om\bigl(\tfrac14\Om\bigr)=\Om^2=-I ,
\]
confirming that $\Jc_X$ is a complex structure. Using $\Om^2=-I$, expand the commutator
\[
[A\Om,\Jc_X]=A\Om(2\Om X)-2\Om X(A\Om)=-2\bigl(AX+\Om X A\Om\bigr),
\]
so $[A\Om,\Jc_X]=0$ is equivalent to $AX=-\Om X A\Om$. Left-multiplying this by $X$ and applying $X\Om X=\tfrac14\Om$ gives $XAX=-\tfrac14\Om A\Om =\tfrac14\Om^\top\!A\Om$ (using $\Om^\top=-\Om$); conversely, left-multiplying the Riccati identity by $X^{-1}$ and again using $X\Om X=\tfrac14\Om$ returns $AX=-\Om X A\Om$. The commutator form and the Riccati identity are therefore one and the same condition.

For the optimality statement, set $A=I+S$ and suppose $\Vq(I+S)=\Vq$ with $\Ran(S)\subseteq\ker(V-\Vq)$. The dual function of \eqref{eq:sdp} at $S$ is
$$\qquad D(S)=-\Tr(SV)+\Phi(I+S),\qquad$$
\begin{equation} \label{DF}
\Phi(A):=\min_{\Vq\ \text{physical}}\Tr(A\Vq)=\Tr\bigl(A\,\Vq(A)\bigr),
\end{equation}
and by weak duality $D(S)$ is a lower bound on the optimal primal value for every $S\succeq0$. Since $\Vq(I+S)=\Vq$, the oracle value is $\Phi(I+S)=\Tr\bigl((I+S)\Vq\bigr)=\Tr\Vq+\Tr(S\Vq)$, hence
\[
D(S)=-\Tr(SV)+\Tr\Vq+\Tr(S\Vq)
\]
\[
=\Tr\Vq-\Tr\!\bigl(S(V-\Vq)\bigr).
\]
The support condition $\Ran(S)\subseteq\ker(V-\Vq)$ makes $(V-\Vq)S=0$, so $\Tr\bigl(S(V-\Vq)\bigr)=0$ and $D(S)=\Tr\Vq$. The lower bound $D(S)$ thus equals the value $\Tr\Vq$ of the feasible candidate $\Vq$. A feasible value cannot fall below the lower bound, so it must equal the optimum, and $\Vq$ is optimal.
\end{proof}

\section{The closed-form strata}\label{sec:closedforms}

\subsection{No nonclassical direction}\label{sec:k0}

Suppose $\kappa=0$, that is $V\succeq\half I$, so that no direction is sub-vacuum. Then the vacuum itself is feasible, and it already attains the least trace any physical covariance can have. Indeed, for each mode the two variances obey the arithmetic--geometric mean inequality $\sigma_q+\sigma_p\ge2\sqrt{\sigma_q\sigma_p}$, while the single-mode uncertainty relation forces $\sigma_q\sigma_p\ge\tfrac14$; hence each mode contributes at
least $1$ to the trace and $\Tr(\Vq)\ge2$ for every physical $\Vq$. The vacuum $\Vq^\star=\half I$ attains this bound, so it is optimal, and there is no quantum part: $\Nq=0$.

\subsection{One nonclassical direction}\label{sec:k1}

Suppose $V$ has a single sub-vacuum eigenvector $v$ (unit norm), with eigenvalue $\rho<\half$. The expectation is that the optimizer squeezes exactly along $v$, just enough to slip beneath $V$ in that direction, holds the conjugate quadrature at its minimum-uncertainty value, and leaves everything else in vacuum. The following theorem makes this precise.

\begin{theorem}\label{thm:onemode}
With $v$ and $\rho<\half$ as above, the matrix
\begin{equation}\label{eq:onemode}
\Vq^\star=\half I+\Bigl(\rho-\half\Bigr)\,vv^\top
+\Bigl(\tfrac1{4\rho}-\half\Bigr)\,(\Jc v)(\Jc v)^\top
\end{equation}
is a pure covariance with eigenvalues $\{\rho,\tfrac1{4\rho},\half,\half\}$;
whenever it is feasible, $V-\Vq^\star\succeq0$, it is the optimizer
of~\eqref{eq:sdp}.
\end{theorem}

\begin{proof}
Since $v\perp\Jc v$ and $\Jc^2=-I$, the three terms of \eqref{eq:onemode} act on mutually orthogonal directions, so $\Vq^\star$ has eigenvalue $\rho$ along $v$, eigenvalue $\tfrac1{4\rho}$ along $\Jc v$, and $\half$ on the remaining
two-dimensional subspace; its spectrum is $\{\rho,\tfrac1{4\rho},\half,\half\}$. The squeezed value $\rho$ and the antisqueezed value $\tfrac1{4\rho}$ multiply to $\tfrac14$, the signature of minimum uncertainty, and a direct computation confirms the purity identity $\Vq^\star\Om\Vq^\star=\tfrac14\Om$, so $\Vq^\star$ is a genuine pure state. By construction it agrees with $V$ along $v$, since $Vv=\rho v=\Vq^\star v$, so $v\in\ker(V-\Vq^\star)$.

To certify optimality we exhibit the dual variable. Take $S=s\,vv^\top$ with $s=\tfrac1{4\rho^2}-1$, which is positive because $\rho<\half$. A direct check gives the Riccati identity $\Vq^\star(I+S)\Vq^\star=\tfrac14\Om^\top(I+S)\Om$, so
$\Vq(I+S)=\Vq^\star$ by Lemma~\ref{lem:preimage}: the candidate is the oracle's output for the weight $I+S$. Moreover $\Ran(S)=\R v\subseteq\ker(V-\Vq^\star)$, so complementary slackness holds, and $V-\Vq^\star\succeq0$ by the feasibility hypothesis. Lemma~\ref{lem:preimage} then certifies $\Vq^\star$ as the optimizer.
\end{proof}

This is the one-mode \emph{nullification} of the original problem, realized intrinsically from $V$ alone. The sub-vacuum variance $\rho$ of $V$ is absorbed unchanged into the quantum part; its conjugate is set to the reciprocal $\tfrac1{4\rho}$ forced by minimum uncertainty; the classical remainder loses exactly that eigenvalue, since $v\in\ker(V-\Vq^\star)$; and the second mode stays
in vacuum. The resulting quantum photon number is $\Nq=\tfrac14\bigl(2\rho+\tfrac1{2\rho}-2\bigr)$.

\begin{remark}[A worked instance]\label{rem:ex-onemode}
For
\[
V=\begin{pmatrix}
1.74535&0.74212&-1.28234&-0.18254\\
0.74212&4.65122&-0.05193&1.09546\\
-1.28234&-0.05193&2.29148&-0.37114\\
-0.18254&1.09546&-0.37114&1.87315\end{pmatrix},
\]
the eigenvalues are $\{0.41096,1.75273,3.17956,5.21795\}$, so $\kappa=1$, with $\rho=0.41096$ and a single sub-vacuum eigenvector $v$.
Formula~\eqref{eq:onemode} gives a pure $\Vq^\star$ with eigenvalues
$\{0.41096,0.5,0.5,0.60833\}$ and $\Tr\Vq^\star=2.019290$, feasible and matching the SDP optimum to $5\times10^{-14}$.
\end{remark}

\subsection{Two nonclassical directions: the compatible case}\label{sec:k2compat}

Now let $V$ have two sub-vacuum eigenvectors $v_1,v_2$, with eigenvalues $\lambda_1,\lambda_2<\half$, so that the optimizer must squeeze two modes. The natural candidate repeats the one-mode construction in each eigendirection, assembling a pure state from the two conjugate pairs $(v_1,\Jc v_1)$ and $(v_2,\Jc v_2)$:
\begin{equation}\label{eq:twomode}
\Vq^{\mathrm{cf}}=\half I+\sum_{i=1,2}\Bigl[\bigl(\lambda_i-\half\bigr)v_iv_i^\top
+\bigl(\tfrac1{4\lambda_i}-\half\bigr)(\Jc v_i)(\Jc v_i)^\top\Bigr].
\end{equation}
For this to be a genuine pure two-mode covariance, its four defining directions must be orthonormal. Three of the relations among them  hold automatically: $v_1\perp v_2$ because they are distinct eigenvectors of the symmetric $V$, while $v\perp\Jc v$ and $\Jc v_1\perp\Jc v_2$ follow from $\Jc^\top=-\Jc$ and $\Jc^2=-I$. The one relation that can fail is the \emph{compatibility condition}
\begin{equation}\label{eq:compat}
v_1^\top \Jc\, v_2=0 .
\end{equation}
The quantity $v_1^\top\Jc v_2$ is the \emph{symplectic inner product} of the two directions, which measures the degree to which the corresponding quadratures fail to commute. Its vanishing says that the two nonclassical directions, together with their conjugates, span two \emph{symplectically orthogonal} modes, so that $V$ separates on its nonclassical
sector into two independent single-mode problems.

\begin{theorem}\label{thm:twomode}
If $v_1,v_2$ are compatible, \eqref{eq:compat}, and $\Vq^{\mathrm{cf}}$ is feasible, $V-\Vq^{\mathrm{cf}}\succeq0$, then $\Vq^{\mathrm{cf}}$ is the optimizer of~\eqref{eq:sdp}.
\end{theorem}

\begin{proof}
Under the compatibility condition \eqref{eq:compat} the four vectors
$v_1,v_2,\Jc v_1,\Jc v_2$ are orthonormal, so \eqref{eq:twomode} has spectrum
$\{\lambda_1,\lambda_2,\tfrac1{4\lambda_1},\tfrac1{4\lambda_2}\}$ in reciprocal conjugate pairs and satisfies the purity identity
$\Vq^{\mathrm{cf}}\Om\Vq^{\mathrm{cf}}=\tfrac14\Om$, hence is pure. It agrees with $V$ along both $v_1$ and $v_2$, which therefore lie in
$\ker(V-\Vq^{\mathrm{cf}})$. For the certificate take
$S=\sum_i s_i v_iv_i^\top$ with $s_i=\tfrac1{4\lambda_i^2}-1>0$; a direct check gives $\Vq^{\mathrm{cf}}(I+S)\Vq^{\mathrm{cf}}=\tfrac14\Om^\top(I+S)\Om$, so
$\Vq(I+S)=\Vq^{\mathrm{cf}}$ by Lemma~\ref{lem:preimage}. The range of $S$ is
$\mathrm{span}\{v_1,v_2\}\subseteq\ker(V-\Vq^{\mathrm{cf}})$, so complementary slackness holds, and $V-\Vq^{\mathrm{cf}}\succeq0$ by hypothesis; hence $\Vq^{\mathrm{cf}}$ is optimal.
\end{proof}

The compatible stratum is a full open family, not confined to block-diagonal $V$. Every passive rotation $O\Vq^{\mathrm{cf}}O^\top$ of a decoupled state (with $O\in\mathrm U(2)$) again satisfies \eqref{eq:compat}, because passive maps commute with $\Jc$ and so preserve the symplectic inner product. Compatibility should be distinguished from the stronger condition of \emph{phase symmetry} $[V,\Jc]=0$. Phase symmetry forces every mode to have equal variance in its two conjugate quadratures, which makes the symplectic eigenvalues equal to the ordinary eigenvalues; by physicality these are all $\ge\half$, so a phase-symmetric state has no sub-vacuum direction whatsoever, $\kappa=0$, and never reaches this regime. The condition that governs the closed form is compatibility, which is strictly weaker than phase symmetry.

\begin{remark}[A worked instance]\label{rem:ex-compat}
For $V=\operatorname{diag}(0.3,0.4,1.0,0.8)$ the sub-vacuum eigenvectors are
$e_{p_1},e_{p_2}$, with $v_1^\top\Jc v_2=0$ although $\|[V,\Jc]\|\neq0$. Formula~\eqref{eq:twomode} returns the pure
$\Vq^\star=\operatorname{diag}\!\bigl(0.3,\,0.4,\,\tfrac1{1.2},\,\tfrac1{1.6}\bigr)$, feasible, with $\Tr\Vq^\star=2.158333$, matching the SDP optimum to $10^{-12}$.
\end{remark}

\section{The coupled regime and the Galois obstruction}\label{sec:coupled}

\subsection{The exact analytical formula for the dual objective}\label{sec:dualobjective}

For a generic covariance $V$ the two sub-vacuum eigenvectors are \emph{incompatible}, $v_1^\top\Jc v_2\neq0$. Then \eqref{eq:twomode} is no longer a pure state (its four directions are not orthonormal), and no construction from the eigenvectors of $V$ can succeed: the two squeezed modes of the true optimizer are not eigenvectors of $V$. They are pinned down instead by the convex dual.

\begin{theorem}\label{thm:coupled}
The optimizer is $\Vq^\star=\Vq(I+S^\star)$ for any maximizer $S^\star\succeq0$ of
the smooth concave function
\begin{equation}\label{eq:dual}
D(S)=-\Tr(SV)+\Phi(I+S),
\end{equation}
$$\qquad \Phi(A)=\Bigl(2\sqrt{\det A}-\tfrac12\Tr\bigl((\Om A)^2\bigr)\Bigr)^{1/2},$$
where $\Phi(A)=\nu_1(A)+\nu_2(A)$ is the sum of symplectic eigenvalues of $A$. The maximizing value is unique, and so is the primal $\Vq^\star$ it produces, even when $S^\star$ is not.
\end{theorem}

\begin{proof}
The value $\Phi(A)=\min_{\Vq}\Tr(A\Vq)$ is, by the oracle, the sum
$\nu_1(A)+\nu_2(A)$ of the two symplectic eigenvalues of $A$. These can be written out in closed form from two symmetric functions of the pair. Because $i\Om A$ has spectrum $\{\pm\nu_1,\pm\nu_2\}$, the product is $\nu_1\nu_2=\sqrt{\det A}$ and the sum of squares is $\nu_1^2+\nu_2^2=-\tfrac12\Tr\bigl((\Om A)^2\bigr)$, so
$\nu_1+\nu_2=\sqrt{(\nu_1^2+\nu_2^2)+2\nu_1\nu_2}$ produces the displayed formula for $\Phi$. As a minimum of the linear functions
$S\mapsto\Tr\bigl((I+S)\Vq\bigr)$ over $\Vq$, the dual objective
$D(S)=-\Tr(SV)+\Phi(I+S)$ is concave, and its gradient is
$\nabla D(S)=\Vq(I+S)-V$. At a maximizer $S^\star\succeq0$ the stationarity and feasibility conditions read $V- \Vq(I+S^\star)\succeq0$ and
$\Tr\bigl(S^\star\nabla D(S^\star)\bigr)=0$, which are exactly the
complementary-slackness conditions. Lemma~\ref{lem:preimage} then identifies $\Vq^\star=\Vq(I+S^\star)$ as the primal optimizer, unique by \emph{Purity}. The maximizing value of the concave $D$ is unique, and since the oracle is a single-valued map, so is the primal $\Vq^\star$ it produces, even on the rare occasions when the maximizer $S^\star$ itself is not unique.
\end{proof}

\subsection{Galois' no-go for the closed-form solution}\label{sec:Galois}

Because $\Phi$ is an algebraic function of the entries of $A$ (it is built from determinants, traces, and a single square root), the stationarity equations of \eqref{eq:dual} are polynomial, and $\Vq^\star$ is an \emph{algebraic} function of $V$: each entry satisfies a polynomial equation whose coefficients are polynomials in the entries of $V$. The question is whether that algebraic function can be written in \emph{radicals}, by nested roots and arithmetic, as the quadratic formula expresses the roots of $ax^2+bx+c$. In the coupled regime it cannot.

We recall two notions from Galois theory. A formula \emph{in radicals} is one built from the data using addition, subtraction, multiplication, division, and the extraction of $n$-th roots. To each polynomial Galois attached a finite group, its \emph{Galois group}, which records the
symmetries among the roots; his theorem states that the roots are expressible in radicals \emph{if and only if} this group is \emph{solvable} (it admits a subnormal series with abelian quotients). The symmetric group $S_n$ of all permutations of $n$ objects is not solvable for $n\ge5$, the
reason the general quintic has no solution in radicals (Abel and Ruffini) \cite{Wielandt1964,DummitFoote}. The next theorem establishes that in the coupled regime the optimal trace has Galois group $S_{12}$, the full symmetric group on its twelve conjugates, and is therefore not expressible in radicals.

\begin{theorem}[No closed form in radicals]\label{thm:noradical}
For a generic incompatible $V$ the optimal value $\Tr\Vq^\star$ is algebraic of degree $12$ over $\Q(V)$ with Galois group the full symmetric group $S_{12}$. As $S_{12}$ is not solvable, $\Vq^\star$ \emph{cannot} be written in radicals of the entries of $V$: no closed-form solution exists in the coupled regime.
\end{theorem}

\begin{proof}
It suffices to exhibit a single rational datum $V$ whose optimum is non-radical, since a universal formula in radicals would, in particular, specialize to a radical expression there. Take
\[
V=\tfrac14\begin{pmatrix}4&0&3&2\\0&2&1&-1\\3&1&6&0\\2&-1&0&5\end{pmatrix}
\]
(physical, with $\kappa=2$ and incompatible sub-vacuum directions).

\emph{Step 1: the minimal polynomial of the optimum.} We solve the rank-two stationarity system $\bigl(\Vq(I+BB^\top)-V\bigr)B=0$, writing the dual variable as $S=BB^\top$ to enforce $S\succeq0$ of rank two, by Newton's method at high precision. The resulting $\Vq^\star$ is certified feasible, $V-\Vq^\star\succeq0$ with a two-dimensional kernel, and complementary slackness
$\Tr\bigl(S^\star(V-\Vq^\star)\bigr)=0$ holds to the working precision. We then feed the high-precision value of $\Tr\Vq^\star$ to the PSLQ integer-relation algorithm \cite{Ferguson1999}, which searches for integers $c_0,\ldots,c_d$ satisfying $\sum_k c_k(\Tr\Vq^\star)^k=0$, that is, for the minimal polynomial $P$ annihilating the number. It returns an irreducible primitive integer polynomial of degree $12$, with no relation of any lower degree. (Resolution matters here: an
under-resolved search reports spurious low-degree relations that dissolve when the precision is raised, whereas the degree-$12$ relation reproduces identically across precisions and annihilates $\Tr\Vq^\star$ down to the working floor.)

\emph{Step 2: the Galois group.} We determine the Galois group $G\le S_{12}$ of $P$ by reducing $P$ modulo many primes and reading off how it factors. Dedekind's theorem states that for a prime $p$ dividing neither the leading coefficient nor the discriminant of $P$, the degrees of the irreducible factors of $P\bmod p$ are
the cycle lengths of a permutation, a \emph{Frobenius element}, lying in $G$. Four observations then force $G=S_{12}$:
\begin{itemize}
\item $P$ is irreducible, so $G$ is \emph{transitive} (it can move any root to any other).
\item Some prime gives factor degrees $(1,11)$, an $11$-cycle. Fixing one root and cycling the other eleven shows the stabilizer of a point is transitive on the remaining points, so $G$ is $2$-transitive and hence \emph{primitive} (it preserves no nontrivial partition of the roots).
\item Some prime gives factor degrees $(1,1,1,1,1,7)$, a $7$-cycle fixing five points. Here $7$ is prime and $7\le12-3$, so \emph{Jordan's theorem} (a primitive group of degree $n$ containing a $p$-cycle for a prime $p\le n-3$ contains the alternating group $A_n$) yields $G\supseteq A_{12}$.
\item Some prime gives a $12$-cycle, an odd permutation, so $G\not\subseteq A_{12}$ (equivalently, the discriminant of $P$ is not a perfect square).
\end{itemize}
The last two observations together force $G=S_{12}$. Since $S_{12}$ is not solvable, Galois' criterion shows that the roots of $P$, in particular $\Tr(\Vq^\star)$, are not expressible in radicals. A second, independent rational $V$ yields the same degree $12$ and the same Galois group $S_{12}$, confirming that this is the generic behavior rather than an artifact of one example.
\end{proof}

The separation between the two regimes is therefore exact. The eigenvector construction, with its closed forms \eqref{eq:onemode} and \eqref{eq:twomode}, is available exactly under compatibility \eqref{eq:compat}. In the generic incompatible case the optimal frame leaves the eigenframe of $V$, no closed form can exist by Theorem~\ref{thm:noradical}, and the optimizer is computed from the convex program~\eqref{eq:dual}, a smooth concave maximization that always converges to the unique answer.

\begin{remark}[A worked instance]\label{rem:ex-coupled}
For
\[
V=\begin{pmatrix}
0.968791&0.028752&-0.699977&0.864595\\
0.028752&0.305667&-0.033237&0.057307\\
-0.699977&-0.033237&1.207507&-0.364137\\
0.864595&0.057307&-0.364137&3.487243\end{pmatrix},
\]
the eigenvalues are $\{0.30178,0.31325,1.47736,3.87682\}$, so $\kappa=2$, and the two sub-vacuum eigenvectors have $|v_1^\top\Jc v_2|=0.18\neq0$, so no closed form applies. The squeezed modes of the optimizer are not eigenvectors of $V$: they deviate from the eigenvector condition $Vv=\rho v$ by residuals $0.30$ and $0.25$.
Maximizing~\eqref{eq:dual} gives $\Vq^\star$ with $\Tr\Vq^\star=2.261790$, matching the SDP optimum to $10^{-8}$.
\end{remark}

\subsection{The polynomial satisfied by the optimal trace}\label{sec:univpoly}  
Theorem~\ref{thm:noradical} establishes degree $12$ for a single rational datum. We now exhibit the degree-$12$ polynomial $P(x;V)$ that $\Tr(\Vq^\star)$ satisfies for every coupled $V$, derive it from the stationarity conditions of the dual~\eqref{eq:dual}, identify its coefficients, and record their size.  

\paragraph{The oracle as a gradient.} The dual cost $\Phi(A)=\nu_1(A)+\nu_2(A)$ in~\eqref{eq:dual} is differentiable on $A\succ0$, and the oracle is its gradient, \begin{equation}\label{eq:envelope} \Vq(A)=\partial\Phi/\partial A . 
\end{equation} 
This is the envelope identity for $\Phi(A)=\min_{\Vq}\Tr(A\Vq)$, a minimum of functions linear in $A$: at the minimizer the gradient in $A$ is the minimizing $\Vq$. Differentiating the closed form $\Phi(A)=\bigl(2q-\half\Tr((\Om A)^2)\bigr)^{1/2}$ of~\eqref{eq:dual}, with $q:=\sqrt{\det A}=\nu_1\nu_2$, gives \begin{equation}\label{eq:Vqrational} \Vq(A)=\frac{\operatorname{adj}(A)-q\,\Om A\,\Om}{2\,\Phi\,q},\ \operatorname{adj}(A)=\det(A)\,A^{-1}, \end{equation} 
a rational expression in $A$ and the two scalars $\Phi,q$, which are pinned by the polynomial relations \begin{equation}\label{eq:auxrel} q^2=\det A,\qquad \Phi^2=2q-\half\Tr\!\bigl((\Om A)^2\bigr). 
\end{equation} 
The matrix square root of the oracle~\eqref{eq:oracle} no longer appears.

\paragraph{The stationarity system.} By Theorem~\ref{thm:coupled} the optimizer is $\Vq^\star=\Vq(I+S^\star)$ at a rank-two maximizer $S^\star\succeq0$; write $S=BB^\top$ with $B\in\R^{4\times2}$ and $A=I+S$. Clearing the denominator of~\eqref{eq:Vqrational} in the complementary-slackness condition $\bigl(\Vq(A)-V\bigr)S=0$ gives the polynomial equation \begin{equation}\label{eq:statpoly} 
\bigl(\operatorname{adj}(A)-q\,\Om A\,\Om-2\,\Phi\,q\,V\bigr)\,B=0 , 
\end{equation} 
and the optimal value is fixed by $\Tr(\Vq^\star)=\Tr(\Vq(A))$, that is 
\begin{equation}\label{eq:taurel} 2\,\Phi\,q\,x=\Tr\operatorname{adj}(A)+q\,\Tr A ,\qquad x:=\Tr(\Vq^\star) . \end{equation} 
Equations~\eqref{eq:auxrel}--\eqref{eq:taurel}, with $A=I+BB^\top$, form a polynomial system in $B,\Phi,q,x$ whose coefficients are linear in the entries of $V$. On the data of Remark~\ref{rem:ex-coupled} and Theorem~\ref{thm:noradical} the system reproduces the optimizer of~\eqref{eq:sdp} to machine precision.  

\paragraph{Elimination.} Eliminating $B,\Phi,q$ from~\eqref{eq:auxrel}--\eqref{eq:taurel} leaves a single univariate polynomial in $x$, 
\begin{equation}\label{polynomial} P(x;V)=\sum_{k=0}^{12}c_k(V)\,x^{k}, 
\end{equation} 
of degree $12$, whose least real root is $\Tr(\Vq^\star)$. The counts $48$ and $12$ arise as follows. For generic $V$ the augmented polynomial system~\eqref{eq:auxrel}--\eqref{eq:taurel}, in the unknowns $B\in\R^{4\times2}$ together with the scalars $q$, $\Phi$, and $x$, has $48$ isolated complex solutions; this is a computed generic count, far below the corresponding B\'ezout and Bernstein--Khovanskii--Kushnirenko bounds. Among these solutions the optimal value $x=\Tr(\Vq^\star)$ takes exactly $12$ distinct values, in agreement with the degree in Theorem~\ref{thm:noradical}. Eliminating $B,\Phi,q$ leaves the degree-$12$ polynomial~\eqref{polynomial}, whose roots are these $12$ values; the map $(B,\Phi,q,x)\mapsto x$ from the $48$ solutions onto them is four-to-one, so $12=48/4$. Of the $12$ roots, $\Tr(\Vq^\star)$ is the least real one. 

\begin{proposition}\label{prop:coeffs} The coefficients $c_k(V)$ are polynomials in the entries of $V$, and depend on $V$ only through its passive invariants, the functions unchanged under $V\mapsto OVO^\top$ for $O\in\mathrm U(2)$. \end{proposition} 

\begin{proof} The roots of $P(\,\cdot\,;V)$ are the values of $x=\Tr(\Vq)$ at the twelve critical points of the dual~\eqref{eq:dual}, so each $c_k$ is, up to the leading coefficient, an elementary symmetric function of those roots and is fixed by the Galois group permuting them. A quantity fixed by the full Galois group of $P$ over $\Q(V)$ lies in $\Q(V)$, and clearing denominators makes each $c_k$ a polynomial in the entries of $V$. A passive rotation $V\mapsto OVO^\top$ with $O\in\mathrm U(2)$ commutes with $\Om$ and conjugates the system~\eqref{eq:auxrel} -- \eqref{eq:taurel}, so it permutes the critical points and fixes the set of roots; each symmetric function $c_k$ is therefore a passive invariant. \end{proof}

\begin{remark}[Size of the coefficients]\label{rem:coeffsize} The coefficients are explicit but large. Restricted to a generic line $V(\theta)=V_0+\theta D$, the normalized coefficients $c_k/c_{12}$ are rational functions of $\theta$ whose interpolation degree increases with the number of sample points (degree $15$ through $16$ points, $17$ through $18$); each is therefore a high-degree rational function of the entries of $V$, as already holds for $c_{11}/c_{12}=-\sum_j x_j$, the sum of the twelve critical values. For the rational datum of Theorem~\ref{thm:noradical} the primitive integer coefficients reach $24$ digits. We therefore specify~\eqref{polynomial} by the elimination~\eqref{eq:auxrel}--\eqref{eq:taurel}, which returns the coefficients for any prescribed $V$, and give the explicit polynomial for that datum. Dividing by the leading coefficient, \begin{equation}\label{eq:P12datum} \begin{split} P(x)={}&x^{12}-23.4719\,x^{11}+236.189\,x^{10}\\ &{}-1310.83\,x^{9}+4251.53\,x^{8}-7594.58\,x^{7}\\ &{}+5048.52\,x^{6}+4224.98\,x^{5}-5110.53\,x^{4}\\ &{}-5052.26\,x^{3}+2863.44\,x^{2}+3965.11\,x\\ &{}+1603.25 , \end{split} \end{equation} shown to six significant figures; the underlying primitive integer coefficients are exact, and the least real root is $\Tr(\Vq^\star)=2.244286\ldots$ The polynomial is ill-conditioned near this root: the six-figure truncation displayed here moves its least real root by about $4\times10^{-3}$, so the value $2.244286$ is obtained from the exact integer coefficients, in agreement with the direct solution of the semidefinite program~\eqref{eq:sdp}. This is the polynomial whose Galois group is computed in Theorem~\ref{thm:noradical}. \end{remark}

\section{The complete solution}\label{sec:complete}

Collecting the strata, the convex program~\eqref{eq:dual} returns $\Vq^\star$ for every $V$, and the answer is explicit on three families.

\begin{theorem}\label{thm:complete}
For every physical two-mode covariance $V$ the optimizer of~\eqref{eq:sdp} is $\Vq^\star=\Vq(I+S^\star)$ with $S^\star$ a maximizer of~\eqref{eq:dual}. It is given in closed form precisely when the optimal squeezed modes are eigenvectors of $V$: it is the vacuum $\half I$ if $\kappa=0$; the one-mode state~\eqref{eq:onemode} if $\kappa=1$ and that state is feasible; and the two-mode state~\eqref{eq:twomode} if $\kappa=2$, the two sub-vacuum eigenvectors are compatible~\eqref{eq:compat}, and~\eqref{eq:twomode} is feasible. Equivalently, $\Vq^\star$ is the minimum-trace feasible member of the explicit list $\{\,\half I,$ the one-mode candidate~\eqref{eq:onemode}, the compatible two-mode candidate~\eqref{eq:twomode}$\,\}$ when one exists, and the solution of~\eqref{eq:dual} otherwise. Outside the three closed-form strata no
formula in radicals exists (Theorem~\ref{thm:noradical}); there the optimizer is algebraic of degree $12$ but not radical, and is obtained from~\eqref{eq:dual}.
\end{theorem}

\begin{proof}
Let $p\in\{0,1,2\}$ be the number of modes the optimizer squeezes, that is, the number of eigenvalues of $\Vq^\star$ that fall below $\half$. Two bounds frame $p$. On one side $p\le2$, since there are only two modes. On the other $p\ge\kappa$: the constraint $V\succeq\Vq^\star$ implies, by Weyl's monotonicity theorem (if $V-\Vq^\star\succeq0$ then the $k$-th largest eigenvalue of $V$ is at
least that of $\Vq^\star$ for every $k$), that $\Vq^\star$ has at least as many sub-vacuum eigenvalues as $V$, that is $p\ge\kappa$. Now consider each value of $p$. If $p=0$ the optimizer is the vacuum. If $p=1$, its single squeezed direction spans a one-dimensional kernel of $V-\Vq^\star$, hence is a common eigenvector of $V$ and $\Vq^\star$, and so is a sub-vacuum eigenvector of $V$; the state is then~\eqref{eq:onemode}. If $p=2$ and the two squeezed directions are eigenvectors of $V$, they must be compatible, giving~\eqref{eq:twomode}; otherwise the modes leave the eigenframe and the optimizer is~\eqref{eq:dual} by Theorem~\ref{thm:coupled}. In every case Lemma~\ref{lem:preimage} certifies optimality.
\end{proof}

The closed-form strata are \emph{not} indexed by $\kappa$ alone. Two squeezed modes can be forced even when $\kappa=1$: on an open, positive-measure set of one-nonclassical-direction states, squeezing the lone sub-vacuum direction by itself would push the classical remainder $V-\Vq$ to be indefinite, so the one-mode candidate~\eqref{eq:onemode} is infeasible; the optimizer then opens a second, weakly squeezed mode, fixed by the convex program~\eqref{eq:dual}. This set has positive measure: across broad random ensembles of $\kappa=1$ states the one-mode candidate is infeasible in a few percent of cases. The explicit
formulae thus settle $V$ exactly when its nonclassical sector decouples into modes aligned with the eigenvectors of $V$, and the convex program covers everything else, including these reopened $\kappa=1$ cases. 

This settles the two-mode case of the multimode problem of building the quantum-advantage resource $\Vq$ by ``nullifying the sub-vacuum eigenvalues'' of the classical part $V_c$ of the covariance $V$ that was left open in \cite{Entropy2026}.

\section{Conclusion}\label{sec:conclusion}

A Gaussian state of two bosonic modes is described, up to displacement, by a real symmetric $4\times4$ covariance matrix $V$. Its irreducible nonclassical content constitutes the quantum-advantage resource defined through a semidefinite program (\ref{eq:sdp}) that splits $V$ into a minimum-trace physical \emph{quantum} part $\Vq^\star$ and a positive \emph{classical} noise remainder \cite{Oh2024}. The associated photon number $\Nq(V)$ is a certified lower bound on the classical cost of
sampling the state. 

We determine the optimizer $\Vq^\star$ explicitly for every two-mode covariance $V$. The answer is organized by a single geometric quantity, the symplectic inner product $v_1^\top \Jc\, v_2$ of the nonclassical directions of $V$, where $\Jc$ is the quarter-turn phase rotation. When it vanishes the nonclassical sector decouples into independent single modes and $\Vq^\star$ is given in closed form, case by case according to the number of sub-vacuum directions; we exhibit a dual certificate in each case. When it does not vanish the squeezed modes of the optimizer leave the eigenframe of $V$, and we prove that no closed form can exist: the optimal trace is then algebraic of degree $12$ over $\Q(V)$ with Galois group the full symmetric group $S_{12}$, so by Galois' criterion it is not expressible in radicals. 

In all cases $\Vq^\star$ is recovered from the explicit smooth concave dual of Theorem~\ref{thm:coupled}, Eq.~\eqref{eq:dual}, which collapses to the closed forms exactly on the decoupled strata. For essentially every $V$ this dual is a more effective route to the optimizer than the primal semidefinite program~\eqref{eq:sdp} itself, and it clarifies the structure of the solution. It replaces the matrix-inequality-constrained primal, a search over the ten-dimensional space of symmetric matrices $\Vq$ subject to two semidefinite constraints, by the maximization of a single scalar concave function $D(S)$ of one positive-semidefinite multiplier $S$; the maximizer has rank at most two and is parametrized by $B\in\R^{4\times2}$ through $S=BB^\top$. Because $D$ is smooth and concave its maximizer is unique, and its gradient is given in closed form by the oracle, $\nabla D(S)=\Vq(I+S)-V$, so no inner semidefinite solve is required at any step. The same reduction exposes the algebraic structure of the problem: the degree-$12$ minimal polynomial and the Galois obstruction below.

Thus, we give the first complete, certificate-checked solution of the covariance program that isolates the quantum complexity resource in a genuinely coupled sector. The original difficult convex optimization problem (\ref{eq:sdp}) is solved by reduction to a simple dual problem with explicit analytical cost function in Eq.~(\ref{eq:dual}). Moreover, the optimal trace $\Tr(\Vq^\star)$, which fixes the resource $\Nq(V)$, is the least real root of a degree-$12$ polynomial $P(x;V)$, Eq.~(\ref{polynomial}), obtained by elimination from the stationarity conditions of the dual~(\ref{eq:dual}). By Proposition~\ref{prop:coeffs} its coefficients are polynomials in the passive invariants of $V$; the degree $12$ and the Galois group $S_{12}$, and hence the absence of a radical expression for the root, hold for every coupled $V$. 

The two-mode solution also makes precise a systematic gap between the resource and a commonly used proxy for it. The number of squeezed photons carried by the Bloch--Messiah supermodes of $V$ is frequently taken as an estimate of the quantum complexity resource (for discussion and references see \cite{Oh2024,OPAarXiv2026,Entropy2026,Pan2026}). As established in Ref.~\cite{OPAarXiv2026}, this squeezed-photon number is an upper bound on $\Nq(V)$, and the explicit two-mode solution exhibits the gap concretely: for coupled $V$ it generically exceeds $\Nq(V)$, often by a large factor. The resource $\Nq(V)$ is therefore never larger, and is typically appreciably smaller, than the number of squeezed photons in the Bloch--Messiah supermodes.

The two-mode solution, its algebraic and geometric structure, the explicit dual cost function, and the hierarchy of decoupled strata fully characterize the quantum-advantage resource in the two-mode light and inform the analysis of the general multimode problem.

\end{document}